\documentclass[12pt]{article}
\usepackage{amsmath}
\usepackage{amsthm}
\usepackage{graphicx}
\usepackage{enumerate}
\usepackage{bbm}
\usepackage{natbib}
\usepackage{url} 
\usepackage{amsfonts}
\usepackage[colorlinks,citecolor=blue,urlcolor=blue]{hyperref}
\usepackage{bm}
\usepackage[flushleft]{threeparttable} %
\usepackage{adjustbox} 
\usepackage{multirow, multicol, makecell, booktabs} %
\newcommand{\RR}{\mathbb{R} }
\newcommand{\NN}{\mathbb{N} }
\newcommand{\PP}{\mathbb{P} }
\newcommand{\EE}{\mathbb{E} }
\newcommand{\blind}{0}

\newtheorem{theorem}{Theorem}[section]
 
\usepackage{natbib}
\usepackage{algorithm}
\usepackage{algpseudocode}
\usepackage{soul}
\usepackage{tabu,multirow}
\usepackage[absolute,overlay]{textpos}
\usepackage{booktabs}
\usepackage{xr}

\allowdisplaybreaks
\begin{document}

\def\spacingset#1{\renewcommand{\baselinestretch}%
{#1}\small\normalsize} \spacingset{1}


\if0\blind
{
  \title{\bf Detecting Structural Shifts and Estimating Change-Points in Interval-Based Time Series}
  \author{Li-Hsien Sun\\
    Graduate Institute of Statistics, National Central University, Taiwan \\
    Zong-Yuan Huang \\
    Graduate Institute of Statistics, National Central University, Taiwan\\
Chi-Yang Chiu\\
Division of Biostatistics, Department of Preventive Medicine,\\ University of Tennessee Health Science Center, Memphis, Tennessee, USA\\
Ning Ning\thanks{
email: \texttt{patning@tamu.edu}}\\
Department of Statistics, Texas A\&M University, Texas, USA}
\date{} 

  \maketitle
} \fi

\if1\blind
{
  \bigskip
  \bigskip
  \bigskip
  \begin{center}
    {\LARGE\bf Title}
\end{center}
  \medskip
} \fi

\bigskip
\begin{abstract}
This paper addresses the open problem of conducting change-point analysis for interval-valued time series data using the maximum likelihood estimation (MLE) framework. Motivated by financial time series, we analyze data that includes daily opening (O), up (U), low (L), and closing (C) values, rather than just a closing value as traditionally used. To tackle this, we propose a fundamental model based on stochastic differential equations, which also serves as a transformation of other widely used models, such as the log-transformed geometric Brownian motion  model. We derive the joint distribution for these interval-valued observations using the reflection principle and Girsanov’s theorem. The MLE is obtained by optimizing the log-likelihood function through first and second-order derivative calculations, utilizing the Newton-Raphson algorithm. We further propose a novel parametric bootstrap method to compute confidence intervals, addressing challenges related to temporal dependency and interval-based data relationships. The performance of the model is evaluated through extensive simulations and real data analysis using S\&P500 returns during the 2022 Russo-Ukrainian War. The results demonstrate that the proposed OULC model consistently outperforms the traditional OC model, offering more accurate and reliable change-point detection and parameter estimates.
\end{abstract}

\noindent%
\vfill


\spacingset{1.75} 
\section{Introduction}
\label{sec:intro}

Change-point analysis plays a vital role in time series analysis, as it allows for the identification of points where key statistical properties, such as the mean, variance, or underlying distribution, experience significant shifts. Recognizing these changes is essential in fields like finance, climatology, and medicine, where structural shifts may indicate critical events like market fluctuations or disease outbreaks. By detecting these changes, analysts can make more informed decisions, enhance predictive modeling, and improve response strategies. The maximum likelihood estimation (MLE) method is highly favored for change-point analysis due to its efficiency and flexibility in parameter estimation. MLE provides consistent and asymptotically normal estimates, making it particularly effective for large sample sizes, which are common in time series data. Furthermore, MLE is adaptable to a wide range of probabilistic models, allowing for a general framework that can be applied across different distributions and model structures. Its robustness and strong theoretical foundations have made it a preferred method for practitioners in change-point detection.

Various studies have advanced change-point analysis using MLE. For example, \cite{PP2005} developed an estimator for the change-point that aligns with the last zero of the binomial CUSUM chart, initially proposed by  \cite{Page1954}. \cite{PS2000,SPC1998} applied MLE to this last zero estimator, further enhancing its utility. \cite{DW2016} explored change detection in multivariate distributions.  \cite{EH2016} combined MLE with the last zero estimator to estimate change points in binomial time series, while \cite{TP2003} extended the MLE approach to the first-order autoregressive model. \cite{ELS2021}  used the Newton-Raphson (NR) algorithm to compute MLEs for change points in binomial time series, demonstrating the method's versatility and reliability in time series analysis.

In practice, time series data can be more comprehensive, especially in financial contexts where multiple observations are recorded daily. For instance, traditional models often focus solely on the daily closing price in financial time series analysis, neglecting other crucial data points such as the daily maximum, minimum, and opening prices, even though these values are typically available.   \cite{BD2003,BD2006} examine methods for evaluating the mean, variance, and covariance, alongside regression analysis using interval-valued observations. In terms of interval time series, a basic approach involves modeling the maximum and minimum values through a vector autoregressive (VAR) model. However, this method can produce unreasonable predictions, such as forecasting a minimum value larger than the maximum. To address this, \cite{ND2007} propose focusing on the center and radius processes rather than the raw minimum and maximum values.

 The VAR model for the first-order difference of the center and radius processes is further discussed in \cite{AGM2011}. 
 \cite{RS2011} extended this work by proposing the center-radius self-exciting threshold autoregressive (CR-SETAR) model. Other relevant research on interval time series includes the works of \cite{GBCC2007,BCCG2008,GLM2015}. Additionally,  \cite{TB2015} introduced the space-time autoregressive (STAR) model, which ensures that the predicted maximum value remains larger than the minimum through parameter constraints. Beyond modeling the raw prices, the range, defined as the difference between logarithmic maximum and minimum prices, has also been studied.  \cite{Chou2005,Chou2006} explored range models driven by geometric Brownian motion (GBM) with stochastic volatility, with estimates obtained via quasi-maximum likelihood.  \cite{CGL2008} proposed a threshold heteroskedastic model for the range driven by the Weibull distribution. More recently,  \cite{LCL2021} analyzed symbolic interval-valued data using auto-interval-regression models.

It remains an open problem how to apply MLE for change-point analysis in interval-valued time series data, and this paper aims to address that challenge. Our first step is to develop a general but fundamental model, which also serves as a transformation of other models, such as the log-transformed GBM model that is widely used in financial modeling. Specifically, we utilize the stochastic differential equation (SDE) in Equation \eqref{eqn:SDE} to define the OULC model, which operates on an interval-based time series. A key difficulty lies in deriving the joint likelihood for this model. By employing the reflection principle and the Girsanov theorem, as described in \cite{Shreve2004} and Equation 1.15.8 in \cite{BS2002}, and further referenced in \cite{CR2013}, we derive the joint distribution for a random vector that includes daily maximum, minimum, opening, and closing values in Theorem \ref{thm:main}. The MLE is then obtained by applying first and second-order derivative optimization rules, after careful calculation of these derivatives and utilizing the NR algorithm. 

Two significant challenges arise when calculating the confidence intervals (CIs) for interval-based time series data: temporal dependence and the intrinsic relationship between interval data points. To overcome these, we propose an innovative method for CI calculation based on the parametric bootstrap approach \citep{ET1993}. We conduct an extensive numerical analysis, simulating various change-point scenarios and performing 1000 replications of each experiment. Furthermore, we apply the proposed method to real stock return data from the S\&P 500 index during the 2022 Russo-Ukrainian War. Both the simulation and empirical data analyses consistently demonstrate that the proposed OULC model outperforms the traditional OC model across all scenarios. By incorporating not only the opening and closing values but also the maximum and minimum, the OULC model offers enhanced parameter and change point estimation, leading to greater accuracy and certainty in results.

The remainder of the paper is organized as follows. Section \ref{MR} provides a comprehensive description of the proposed model and methodology. In Section \ref{SS}, the performance of the method is demonstrated through simulation studies. Section \ref{sec:ES} presents an application of the method by analyzing daily returns of the S\&P500 during the Russo-Ukrainian War in 2022. Finally, conclusions are drawn in Section \ref{CON}. The Supplement includes additional numerical experiments, detailed derivations of the first and second-order derivatives of the log-likelihood function, and tables presenting comprehensive numerical results.

\section{Main result}\label{MR}
In this section, we propose the OULC model in Section \ref{PM}, describe the likelihood function in Section \ref{Likelihood}, detail the maximum likelihood estimation in Section \ref{MLE}, and present the confidence interval construction in Section \ref{sec:CI}.

\subsection{The OULC model}\label{PM}
We propose the OULC model, which operates on an interval-based time series on the $t_i$-th day as $X_i=(O_i,U_i,L_i,C_i)$ for $i\in \{1,\ldots,n\}$ with $n$ being the termial time. 
Here, $O_i$ stands for the opening (starting) value $Y(t_{i-1})$, $U_i$ stands for the upper (highest) value $U_i=\mbox{max}_{t_{i-1}\leq t\leq t_i}Y(t)$, $L_i$ stands for the lowest value $L_i=\mbox{min}_{t_{i-1}\leq t\leq t_i}Y(t)$, and $C_i$ stands for the closing value $Y(t_{i})$. We assume that $Y(t)$ is driven by a SDE in the basic form
\begin{eqnarray}
	\label{eqn:SDE}
	Y(t_i)=Y(t_{i-1})+\int_{t_{i-1}}^{t_{i}}\mu dt+\int_{t_{i-1}}^{t_{i}}\sigma dW(t),
	\label{eq:model.first}
\end{eqnarray}
for $t_{i-1}\leq t\leq t_i$, 
where $W(t)$ is a standard Brownian motion, $\mu$ is the drift parameter, and $\sigma^2$ is the volatility parameter. 

The proposed OULC model has a natural interpretation in financial time series. Google Finance typically provides a comprehensive set of financial data for individual stocks. Here is a breakdown of the common financial data available: opening price (the price at which a stock starts trading when the market opens on a particular day), highest price (the maximum price at which the stock traded during the day), lowest price (the minimum price at which the stock traded during the day), and closing price (the last price at which the stock traded during the regular trading hours on a particular day). SDE \eqref{eq:model.first} and its variants are widely used to model stock prices and other financial instruments due to their ability to capture the continuous and random nature of market movements. Now we further illustrate how SDE \eqref{eq:model.first} relates to a GBM which is an important example used in mathematical finance to model stock prices. Consider the intra-daily stock price satisfying the GBM
$$
S(t_{i})=S(t_{i-1})+\int_{t_{i-1}}^{t_{i}}\mu' S(t) dt+\int_{t_{i-1}}^{t_{i}}\sigma S(t)dW(t).
$$  
The intra-daily log price $Y(t)=\log(S(t))$, which is known as stock return, evolves according to SDE \eqref{eq:model.first} with $\mu=\mu'-\frac{\sigma^2}{2}$. We refer to \cite{AMZ2005,ABDE2001} for further details.

\subsection{Likelihood function}\label{Likelihood}
We aim to obtain the likelihood function for the OULC model and its corresponding profile MLE in this subsection. First, let us recall the two-stage change-point problem formulation \citep{ELS2021,SPC1998}. Let \{$X_i:i = 1,2,\ldots,n$\} be a time series where $n$ denotes the terminal time. Given the parameter space 
\begin{align*}
		\Theta = \Big\{(\gamma_0,\gamma_1, \tau)\,\big|\,\gamma_0 \in \RR^{n_0},\gamma_1\in\RR^{n_1}, \tau \in \NN\Big\},
\end{align*}
where $\gamma_0$ and $\gamma_1$ are parameters for two general parametric marginal density functions $f_{\gamma_0}$ and $f_{\gamma_1}$, respectively, in a two-stage data structure:
\begin{align*}
	\begin{gathered}
		X_1, X_2, \ldots, X_{\tau} \sim f_{\gamma_0} \longrightarrow \mbox{data before structure change} \\  
		X_{\tau+1}, X_{\tau+2}, \ldots, X_n \sim f_{\gamma_1} \longrightarrow \mbox{data after structure change}
	\end{gathered}
\end{align*}
with $\tau$ representing the change-point. The general form of the log-likelihood for the two-stage data structure is expressed as
\begin{align}
	&\ell(\gamma_0, \gamma_1, \tau) = \sum_{i=1}^{\tau} \log(f_{\gamma_0}(x_i)) + \sum_{i=\tau+1}^n \log(f_{\gamma_1}(x_i)). 
	\label{general_log-likelihood}
\end{align}

Now, for the proposed OULC model, the parameter space  is defined as 
\begin{align*}
	\begin{gathered}
		\Theta = \Big\{(\mu_0,\mu_1,\sigma^2_0,\sigma^2_1,\tau)\,\Big|\,\mu_0,\mu_1 \in (-\infty,\infty), \,\sigma_0,\sigma_1 \in (0,\infty), \,
		\tau \in \{3,4,\ldots,n-3\}\Big\}.
	\end{gathered}
\end{align*}
That is, to avoid non-identiability, we assume that there are at least three time points before and after the change point. 
We then construct the log-likelihood function for the OULC model. 

\begin{theorem}
	\label{thm:main}
Given data $\Vec{o}=(o_1,\cdots,o_n)$, $\Vec{u}=(u_1,\cdots,u_n)$, $\Vec{l}=(l_1,\cdots,l_n)$, and $\Vec{c}=(c_1,\cdots,c_n)$,
the log-likelihood with one change point $\tau$ is written as
\begin{align}
	&\ell (\boldsymbol{\theta}) = \sum_{i=1}^{\tau} \log(f_{\mu_0,\sigma^2_0}(u_i,l_i,c_i\,|\,o_i)) + \sum_{i=\tau+1}^n \log(f_{\mu_1,\sigma^2_1}(u_i,l_i,c_i\,|\,o_i)), 
	\end{align}
where $\boldsymbol{\theta} = \{\mu_0, \mu_1, \sigma^2_0, \sigma^2_1,  \tau \}\in \Theta$. 
The functions $f_{\mu_0,\sigma^2_0}$ and $f_{\mu_1,\sigma^2_1}$ are obtained by substituting $(\mu,\sigma^2)$ in $f_{\mu,\sigma^2}$ with $(\mu_0, \sigma^2_0)$ and $(\mu_1, \sigma^2_1)$, respectively.
Here, $f_{\mu,\sigma^2}$ is given by
\begin{align}
		f_{\mu,\sigma^2}(u_i,l_i,c_i\,|\,o_i)=&\sum_{k=-\infty}^{\infty} g^{(1)}_{\sigma^2}(k,u_i,l_i,c_i\,|\,o_i) h^{(1)}_{\mu,\sigma^2}(k,u_i,l_i,c_i\,|\,o_i)\label{f_4}\\
		&
		-\sum_{k=-\infty}^{\infty} g^{(2)}_{\sigma^2}(k,u_i,l_i,c_i\,|\,o_i)h^{(2)}_{\mu,\sigma^2}(k,u_i,l_i,c_i\,|\,o_i),	\nonumber	
	\end{align}
	where 
	\begin{eqnarray}
\nonumber		g^{(1)}_{\sigma^2}(k,u_i,l_i,c_i\,|\,o_i) &=&\frac{4k(k+1)}{\sqrt{2\pi} \sigma^3} \left(1-\frac{(c_i+o_i-2u_i-2k(u_i-l_i))^2}{\sigma^2} \right),\\
	g^{(2)}_{\sigma^2}(k,u_i,l_i,c_i\,|\,o_i)&=&\frac{4k^2}{\sqrt{2\pi} \sigma^3} \left(1-\frac{(c_i-o_i-2k(u_i-l_i))^2}{\sigma^2} \right),	\label{f_5}\\
\nonumber		h^{(1)}_{\mu,\sigma^2}(k,u_i,l_i,c_i\,|\,o_i) &=& \exp\left\{ -\frac{(c_i+o_i-2u_i-2k(u_i-l_i))^2}{2\sigma^2} -\frac{\mu^2}{2\sigma^2}+\frac{\mu(c_i-o_i)}{\sigma^2} \right\}, \\
		h^{(2)}_{\mu,\sigma^2}(k,u_i,l_i,c_i\,|\,o_i) &=& \exp\left\{ -\frac{(c_i-o_i-2k(u_i-l_i))^2}{2\sigma^2} -\frac{\mu^2}{2\sigma^2}+\frac{\mu(c_i-o_i)}{\sigma^2} \right\}.\nonumber	
	\end{eqnarray}
\end{theorem}
\begin{proof}
Without loss of generality, the SDE \eqref{eq:model.first} can be written as
\begin{eqnarray*}
	Y(t)=Y(0)+\int_0^t\mu d\tilde{t}+\int_0^t\sigma dW(\tilde{t}). 
\end{eqnarray*}
Given $M(t)=\mbox{max}_{0\leq \tilde{t}\leq t}Y(\tilde{t})$ and $m(t)=\mbox{min}_{0\leq \tilde{t}\leq T}Y(\tilde{t})$, we have  
\begin{align*}
&\hspace{-0.5cm}\PP\left(a\leq m(t) \leq M(t) \leq b, Y(t)\in dy \,\Big|\,Y(0)=y_0\right)\\
=&\EE\Big[\mathbbm{1}_{\{a\leq m(t) \leq M(t) \leq b, Y(t)\in dy\}}\,\Big|\,Y(0)=y_0\Big]\\
=&\EE\Big[\mathbbm{1}_{\{\frac{a-y}{\sigma}\leq \frac{m(t)-y}{\sigma}\leq \frac{M(t)-y}{\sigma}\leq \frac{b-o}{\sigma}, \frac{Y(t)-y}{\sigma}\in dy\}}\,\Big|\,Y(0)=y_0\Big]\\
=&\frac{1}{\sqrt{2\pi t}}\exp\left\{   -\frac{\mu^2}{2\sigma^2}+\frac{\mu(y-y_0)}{\sigma^2} \right\}\sum_{k=-\infty}^\infty\exp\left\{ -\frac{(y-y_0-2k(b-a))^2}{2\sigma^2}  \right\}\\
&-\frac{1}{\sqrt{2\pi t}}\exp\left\{   -\frac{\mu^2}{2\sigma^2}+\frac{\mu(y-y_0)}{\sigma^2} \right\}\sum_{k=-\infty}^\infty\exp\left\{ -\frac{(y+y_0-2b-2k(b-a))^2}{2\sigma^2}  \right\}.
\end{align*}
The last equation follows from \cite{CR2013} and the Girsanov Theorem in \cite{Shreve2004}; see also Equation 1.15.8 on page 227 in \cite{BS2002}, the theorem on page 26 in \cite{Freedman1971}, and Proposition 8.10 in \cite{KS1998} for related results on the joint distribution of the maximum, minimum, and terminal values of standard Brownian motion given the initial value. By differentiating with respect to $a$ and $b$, we obtain
\begin{align}
\nonumber f_ {\mu,\sigma^2}(b,a,y\,|\,y_0)=&\sum_{k=-\infty}^{\infty} g^{(1)}_{\sigma^2}(k,b,a,y\,|\,y_0) h^{(1)}_{\mu,\sigma^2}(k,b,a,y\,|\,y_0)\\
		&
		-\sum_{k=-\infty}^{\infty} g^{(2)}_{\sigma^2}(k,b,a,y\,|\,y_0)h^{(2)}_{\mu,\sigma^2}(k,b,a,y\,|\,y_0),
		\label{OULC_density}
\end{align}
with $a\leq y_0,y\leq b$, where $g^{(i)}_{\sigma^2}$ and $h^{(i)}_{\mu,\sigma^2}$ for $i=1,2$ are given by equation \eqref{f_5}. 
Hence, given $\Vec{o}=(o_1,\cdots,o_n)$, $\Vec{u}=(u_1,\cdots,u_n)$, $\Vec{l}=(l_1,\cdots,l_n)$, and $\Vec{c}=(c_1,\cdots,c_n)$, based on equation \eqref{OULC_density}, the joint distribution of $U,L,C$ given $O$ denoted as $f_{\mu,\sigma^2}(u_i,l_i,c_i\,|\,o_i)$ is given by equation \eqref{f_4}.

Furthermore, the corresponding likelihood function without structure change is  
\begin{align*}
		&\displaystyle L(\mu,\sigma^2\,|\,\Vec{u},\Vec{l},\Vec{o},\Vec{c})\\
		&=
		  f_{\mu,\sigma^2}(u_n,l_n,c_n,o_n\,|\,u_{n-1},l_{n-1},c_{n-1},o_{n-1},\cdots,u_{1},l_{1},c_{1},o_{1})\times\cdots\times f_{\mu,\sigma^2}(u_{1},l_{1},c_{1},o_{1}),
	\end{align*}
	with $l_i \leq c_i,o_i \leq u_i$, for $i=1,\cdots,n$. 
Applying the Markovian property of the logarithm-transformed stochastic process and the assumption on $o_i=c_{i-1}$, gives 
	 \begin{align*}
 &\hspace{-0.5cm}f_{\mu,\sigma^2}(u_i,l_i,c_i,o_i\,|\,u_{i-1},l_{i-1},c_{i-1},o_{i-1},\cdots,u_{1},l_{1},c_{1},o_{1})\\
= & f_{\mu,\sigma^2}(u_i,l_i,c_i,o_i\,|\,u_{i-1},l_{i-1},c_{i-1},o_{i-1})\\
= & f_{\mu,\sigma^2}(u_i,l_i,c_i\,|\,o_i, u_{i-1},l_{i-1},o_{i-1})\\
=& f_{\mu,\sigma^2}(u_i,l_i,c_i\,|\,o_i). 
\end{align*} 
Finally, through the general form \eqref{general_log-likelihood}, the corresponding likelihood function with one change point is  
\begin{align}
	&\ell (\boldsymbol{\theta}) = \sum_{i=1}^{\tau} \log(f_{\mu_0,\sigma^2_0}(u_i,l_i,c_i\,|\,o_i)) + \sum_{i=\tau+1}^n \log(f_{\mu_1,\sigma^2_1}(u_i,l_i,c_i\,|\,o_i)), 
	\end{align}
with  $l_i \leq c_i,o_i \leq u_i$ for $i=1,\cdots,n$, where $f_{\mu_0,\sigma^2_0}$ and $f_{\mu_1,\sigma^2_1}$ are given by equation \eqref{f_4}. The proof is complete. 	
\end{proof}

\subsection{Maximum Likelihood Estimation}\label{MLE}
Write the maximum likelihood estimators, given $\tau$, as 
\begin{align}
	(\widehat{\mu}_0(\tau), \widehat{\mu}_1(\tau), \widehat{\sigma}_0^{2}(\tau), \widehat{\sigma}_1^{2}(\tau))= \mathop{\mbox{argmax}}\limits_{(\mu_0,\mu_1,\sigma_0,\sigma_1)} \ \ell(\mu_0, \mu_1, \sigma^2_0, \sigma^2_1\,|\,\tau),
\end{align}
where $\ell(\mu_0, \mu_1, \sigma^2_0, \sigma^2_1\,|\,\tau)$ is $\ell (\boldsymbol{\theta})$ conditional on a given $\tau$. Our methodology proceeds in the following three steps:

\noindent \textbf{Step 1}. 
The first order condition gives $\widehat{\mu}_0 $ and $\widehat{\mu}_1$ as below:
\begin{align*}
	\widehat{\mu}_0(\tau) = \frac{1}{\tau} \sum^{\tau}_{i=1} (c_i - o_i)\quad\text{and}\quad
	\widehat{\mu}_1(\tau) = \frac{1}{n - \tau} \sum^{n}_{i = \tau+1} (c_i - o_i). 
\end{align*}

\noindent \textbf{Step 2}. 
$\widehat{\sigma}_0^{2}$ and $ \widehat{\sigma}_1^{2}$ should satisfy
$$ \partial_{\sigma^2_0}\ell \,|\,_{\mu_0=\widehat{\mu}_0,\sigma^2_0=\widehat{\sigma}_0^{2}}=0\quad \text{and}\quad  \partial_{\sigma^2_0}\ell \,|\,_{\mu_1=\widehat{\mu}_1,\sigma^2_1=\widehat{\sigma}_1^{2}}=0,$$ 
where 
\begin{align} 
	\partial_{\sigma^2_0}\ell&=\sum^{\tau}_{i=1}\Bigg\{\frac{\sum_{k=-\infty}^{\infty}\left [ \left(\partial_{\sigma_0^2}g_{\sigma^2_0}^{(1)}\right)  h_{\mu_0,\sigma_0^2}^{(1)}(k,u_i,l_i,c_i\,|\,o_i) +\left(\partial_{\sigma_0^2}h_{\mu_0,\sigma_0^2}^{(1)}\right)g_{\sigma^2_0}^{(1)}(k,u_i,l_i,c_i\,|\,o_i)  \right ]}{  \sum_{k=-\infty}^{\infty}g_{\sigma^2_0}^{(1)} h_{\mu_0,\sigma_0^2}^{(1)}(k,u_i,l_i,c_i\,|\,o_i) -  \sum_{k=-\infty}^{\infty}g_{\sigma^2_0}^{(2)} h_{\mu_0,\sigma_0^2}^{(2)}(k,u_i,l_i,c_i\,|\,o_i)   }  \nonumber\\
	&\hspace{1.5cm}-\frac{\sum_{k=-\infty}^{\infty}\left [ \left({\partial_{\sigma_0^2} }g_{\sigma^2_0}^{(2)}\right)h_{\mu_0,\sigma_0^2}^{(2)}(k,u_i,l_i,c_i\,|\,o_i) +\left(\partial_{\sigma_0^2}h_{\mu_0,\sigma_0^2}^{(2)}\right)g_{\sigma^2_0}^{(2)}(k,u_i,l_i,c_i\,|\,o_i)   \right ]}{ \sum_{k=-\infty}^{\infty}g_{\sigma^2_0}^{(1)} h_{\mu_0,\sigma_0^2}^{(1)} (k,u_i,l_i,c_i\,|\,o_i)-\sum_{k=-\infty}^{\infty}g_{\sigma^2_0}^{(2)} h_{\mu_0,\sigma_0^2}^{(2)}(k,u_i,l_i,c_i\,|\,o_i)  }\Bigg\}\label{eqn_2}\\    \partial_{\sigma^2_1}\ell&=\sum^{n}_{i=\tau+1}\Bigg\{\frac{\sum_{k=-\infty}^{\infty}\left [ \left(\partial_{\sigma_1^2}g_{\sigma^2_1}^{(1)}\right)h_{\mu_1,\sigma_1^2}^{(1)}(k,u_i,l_i,c_i\,|\,o_i) +\left(\partial_{\sigma_1^2}h_{\mu_1,\sigma_1^2}^{(1)}\right)g_{\sigma^2_1}^{(1)}(k,u_i,l_i,c_i\,|\,o_i)  \right ]}{  \sum_{k=-\infty}^{\infty}g_{\sigma^2_1}^{(1)} h_{\mu_1,\sigma_1^2}^{(1)}(k,u_i,l_i,c_i\,|\,o_i)   -   \sum_{k=-\infty}^{\infty}g_{\sigma^2_1}^{(2)} h_{\mu_1,\sigma_1^2}^{(2)} }\nonumber\\
	&\hspace{1.7cm} -\sum_{k=-\infty}^{\infty}\frac{\left [ \left({\partial_{\sigma_1^2} }g_{\sigma^2_1}^{(2)}\right)h_{\mu_1,\sigma_1^2}^{(2)}(k,u_i,l_i,c_i\,|\,o_i) +\left(\partial_{\sigma_1^2}h_{\mu_1,\sigma_1^2}^{(2)}\right)g_{\sigma^2_1}^{(2)}(k,u_i,l_i,c_i\,|\,o_i)   \right ]}{  \sum_{k=-\infty}^{\infty}g_{\sigma^2_1}^{(1)} h_{\mu_1,\sigma_1^2}^{(1)} (k,u_i,l_i,c_i\,|\,o_i)  -   \sum_{k=-\infty}^{\infty}g_{\sigma^2_1}^{(2)} h_{\mu_1,\sigma_1^2}^{(2)}(k,u_i,l_i,c_i\,|\,o_i)  }\Bigg\}.\nonumber
\end{align}

However, the MLE for $(\sigma_0^2,\sigma_1^2)$ cannot be solved explicitly, so we resort to the NR algorithm. The NR algorithm is an iterative numerical method used to find approximate solutions to real-valued functions, particularly for finding the roots (or zeroes) of a function by linearizing it around an initial guess. Since sensitivity to the initial value is a significant issue in the NR algorithm, several different initial values are tested as suggested by \cite{Knight2000}. Additionally, given that $\sigma_0$ and $\sigma_1$ are constrained parameters, we apply the log-transformation $$\zeta_0=\log(\sigma_0)\qquad\text{and}\qquad\zeta_1=\log(\sigma_1)$$  as recommended by \cite{Mac2014}, so that $\zeta_0,\zeta_1\in (-\infty,\infty)$. The profile log-likelihood function with a given $\tau$, after the transformation, is written as:
\begin{align*}
	\widetilde{\ell}(\mu_0, \mu_1, \zeta_0, \zeta_1\,|\,\tau) = \ell(\mu_0, \mu_1, \sigma^2_0, \sigma^2_1\,|\,\tau).
\end{align*}
Denote the profile MLE of the transformed parameters as
\begin{align*}
	(\widehat{\mu}_0(\tau), \widehat{\mu}_1(\tau), \widehat{\zeta}_0(\tau), \widehat{\zeta}_1(\tau)) = \mathop{\mbox{argmax}}\limits_{\mu_0,\mu_1,\zeta_0,\zeta_1 \in (-\infty, \infty)^4} \ \widetilde{\ell}(\mu_0(\tau),\mu_1(\tau), \zeta_0(\tau), \zeta_1(\tau)\,|\,\tau).
\end{align*}
Given $\tau$, ${\widehat\mu}_0$ and ${\widehat\mu}_1$, the pseudo-algorithm for the NR method to find $\widehat{\zeta}_0$ and $\widehat{\zeta}_1$, consequently $\widehat{\sigma}_0^{2}$ and $\widehat{\sigma}_1^{2}$, is provided in Algorithm \ref{NR-algo}. 
\begin{algorithm}[t!]
	\caption{NR algorithm in this setting}\label{NR-algo}
	\vspace{0.1cm}
	\textbf{Input:} The threshold value $\bm{\epsilon}=(\epsilon_{\zeta_{0}},\epsilon_{\zeta_{1}})$ and the initial value $\bm{\Lambda}^{(0)}=({\zeta}_0^{(0)},{\zeta}_1^{(0)}) $\\
	\textbf{Iterate:} For $k=0,1,....$
	\[
	\begin{bmatrix}
		\zeta_0^{(k+1)} \\ \zeta_1^{(k+1)}  \\
	\end{bmatrix}
	=
	\begin{bmatrix}
		\zeta_0^{(k)} \\ \zeta_1^{(k)}  \\
	\end{bmatrix}
	-
	[H(\bm{\Lambda}^{(k)})]^{-1}
	G(\bm{\Lambda}^{(k)})    	
	,       	
	\] 
	\hspace{1.5cm} where
	\begin{align*}
		H(\bm{\Lambda}^{(k)})
		= -\begin{bmatrix}
			\frac{\partial^2 \widetilde{\ell}}{\partial \zeta_0^{(k)} \partial \zeta_0^{(k)}} &
			\frac{\partial^2 \widetilde{\ell}}{\partial \zeta_0^{(k)} \partial \zeta_1^{(k)}}\vspace{0.1cm} \\
			\frac{\partial^2 \widetilde{\ell}}{\partial \zeta_1^{(k)} \partial \zeta_0^{(k)}} &
			\frac{\partial^2 \widetilde{\ell}}{\partial \zeta_1^{(k)} \partial \zeta_1^{(k)}} 
		\end{bmatrix}
		\quad
		\text{and}
		\quad
		G(\bm{\Lambda}^{(k)})=
		\begin{bmatrix}
			\partial_{\zeta_0^{(k)}} \widetilde{\ell},\, \partial_{\zeta_1^{(k)}} \widetilde{\ell}
		\end{bmatrix}^{T}.
	\end{align*}
	
 \hspace{1.5cm} Stop, if $|\bm{\Lambda}^{(k+1)}-\bm{\Lambda}^{(k)}|< \bm{\epsilon} $.\\
	\textbf{Output:} $\widehat{\bm{\Lambda}}={\bm{\Lambda}}^{(k+1)}$\\
\end{algorithm}

\noindent \textbf{Step 3}. 
The change-point that maximizes the profile log-likelihood is obtained by
\begin{align*}
	\widehat{\tau} = \mathop{\mbox{argmax}}\limits_{\tau \in \{3,4,\ldots,n-3\}}\widetilde{\ell}(\widehat{\mu}_0(\tau), \widehat{\mu}_1(\tau), \widehat{\zeta}_0(\tau), \widehat{\zeta}_1 (\tau),\tau ).
\end{align*}
At last, we obtain the profile MLE
$$\boldsymbol{\widehat\theta} = (\widehat\mu_0,\widehat\mu_1,\widehat \zeta_0,\widehat \zeta_1,\widehat\tau):=(\widehat{\mu}_0(\widehat\tau), \widehat{\mu}_1(\widehat\tau), \widehat{\zeta}_0(\widehat\tau), \widehat{\zeta}_1(\widehat\tau),\widehat\tau).$$

\subsection{Confidence interval construction}
\label{sec:CI}
We propose a new methodology for constructing confidence intervals (CIs), which is based on the parametric bootstrap method \citep{ET1993}, as follows.

\noindent \textbf{Step 1}. Obtain the MLE $(\widehat{\tau},\widehat{\mu}_0 , \widehat{\mu}_1, \widehat{\sigma}_0^{2} , \widehat{\sigma}_1^{2})$, where $\widehat{\sigma}_0=\exp(\widehat \zeta_0)$ and $\widehat{\sigma}_1=\exp(\widehat \zeta_1)$.

\noindent \textbf{Step 2}. Given the initial value $o_1$, generate the bootstrap samples $\Vec{o}(b)=(o_2(b),\cdots,o_n(b))$, $\Vec{u}(b)=(u_1(b),\cdots,u_n(b))$, $\Vec{l}=(l_1(b),\cdots,l_n(b))$, and $\Vec{c}(b)=(c_1(b),\cdots,c_n(b))$ using the model given by equation \eqref{eq:model.first} with the parameters $(\widehat{\mu}_0,\widehat{\sigma}^2_0)$ before $\widehat\tau$ and  $(\widehat{\mu}_1,\widehat{\sigma}^2_1)$ after $\widehat\tau$ respectively. 

\noindent \textbf{Step 3}. Obtain the MLE $({\widehat{\tau}}(b),{\widehat{\mu}}_0(b) , {\widehat{\mu}}_1(b), {\widehat{\sigma}_0^{2}(b)} , {\widehat{\sigma}_1^{2}(b)} )$.

\noindent \textbf{Step 4}. Repeat Step 2 and Step 3 $B$ times and obtain the set of the bootstrap MLEs $({\widehat{\tau}}(b),{\widehat{\mu}}(b)_0 , {\widehat{\mu}}_1(b), {\widehat{\sigma}_0^{2}(b)}  , {\widehat{\sigma}_1^{2}(b)} )$ for $b=1,\cdots,B$, where $B$ is the number of bootstrap replications. 

\noindent \textbf{Step 5}. 	Define the CI for $\mu_0$ as the 
$[\alpha B/2]$-th 
value and the $[B(1-\alpha/2)]$-th value in the ordered values of  
${\widehat{\mu}}_0(b)$ for $b=1,\cdots,B$. Similarly, define the CIs for  ${\mu}_1$, $\sigma_0^2$, and $\sigma_1^2$.

\noindent \textbf{Step 6}. Referred to \cite{TP1997}, define the CI for $\tau$ as the highest density (frequency) of the bootstrap change point estimates as follows:
\begin{enumerate}[(6.1)]
	\item Denote the set $\{c_1,\cdots,c_{\widetilde{N}_B}\}$ as the distinct values of $\{{\widehat{\tau}}(b):b=1,\cdots,B\}$, where $\widetilde{N}_B\leq B$.
	\item Let $\{c_{(1)},\cdots,c_{(\widetilde{N}_B)}\}$ be the order values of the set $\{c_1,\cdots,c_{\widetilde{N}_B}\}$. The confidence set is the minimal subset $ \{c_{(1)},\cdots,c_{(\widetilde{N}_{\tau})}\} \in\{c_{(1)},\cdots,c_{(\widetilde{N}_B)}\}$ such that 
	\[
	\sum_{b=1}^B\frac{{\mathbb{I}}_{ \left\{\widehat{\tau}(b) \in \{c_{(1)},\cdots,c_{(\widetilde{N}_{\tau})}\}\right\}  }}{B}\geq (1-\alpha). 
	\]   
\end{enumerate}

\section{Numerical analysis}\label{SS}
In this section, we conduct extensive simulations to evaluate the performance of the proposed method. Section \ref{sec:Description} provides a description of the data. Structural changes based on $\sigma^2$ are discussed in Section \ref{sec:Structure_change_sigma}, and structural changes based on $\mu$ are presented in Section \ref{sec:Structure_change_mu}.

\subsection{Description}
\label{sec:Description}
We set the terminal time as \( n = 250 \), representing approximately one year of daily (financial) data. To ensure statistical robustness, we examined various change point scenarios \( \tau = n/2, n/3, n/5, \) and \( n/10 \), conducting \( R = 1000 \) replications for each scenario. The threshold values for the NR algorithm were set to \( 10^{-6} \) across all cases. The change point $\tau$ and the model parameters $(\mu_0,\mu_1,\sigma_0,\sigma_1)$ were estimated using observations $X_i=(O_i,U_i,L_i,C_i)$ for $i\in \{1,\ldots,n\}$. We compare the proposed approach applied to the OULC model with the traditional method applied to the OC model, which estimates parameters using only the opening and closing values. Since the OC model is a simpler special case of the OULC model, referred to \cite{LS2019}, the log-likelihood function with one change point is given by
\begin{align*}
	\ell(\boldsymbol{\theta} \,|\, \Vec{o}, \Vec{c} \,) 
	= &-\frac{\tau}{2} \log(2 \pi \sigma_0^2) - \frac{1}{2 \sigma_0^2} \sum^{\tau}_{i = 1} (c_i - o_i - \mu_0 )^2 \\
	&-\frac{n - \tau}{2} \log(2 \pi \sigma_1^2) - \frac{1}{2 \sigma_1^2} \sum^{n}_{i = \tau+ 1} (c_i - o_i - \mu_1 )^2, 
\end{align*}
and the corresponding profile likelihood estimators are
\[
\widehat{\mu}_0(\tau) = \frac{1}{\tau} \sum^{\tau}_{i = 1} (c_i - o_i),\qquad  \widehat{\mu}_1(\tau) = \frac{1}{n - \tau} \sum^{n}_{i = \tau+1} (c_i - o_i),
\] 
\[
\widehat{\sigma}_0^2(\tau) = \frac{1}{\tau} \sum^{\tau}_{i = 1} (c_i - o_i - \widehat{\mu}_0)^2\quad \text{and}\quad \widehat{\sigma}_1^2(\tau) = \frac{1}{n - \tau} \sum^{n}_{i = \tau + 1} (c_i - o_i - \widehat{\mu}_1)^2. 
\]
To assess the performance of the estimators, we use the root mean square error (RMSE) and root error (RE), defined respectively as
\[
{\text{RMSE}}(\widehat\theta)=\sqrt{\frac{1}{R}\sum_{i=1}^R(\widehat\theta_i-\theta)^2}
\qquad \text{and}\qquad 
{\text{RE}}(\widehat\theta)=\frac{{\text{RMSE}}}{\text{True value}},
\]
where $\theta= (\mu_0,\mu_1,\sigma_0,\sigma_1,\tau)$.
See \cite{HC1996} for instance. 

\subsection{Structure change based on $\sigma^2$}
\label{sec:Structure_change_sigma}
We examined the model performance under structural changes in $\sigma^2$. Specifically, we firstly considered a change-point scenario with $\mu_0 = \mu_1 = 0.0008$, $\sigma_0^2 = 0.000169$, and $\sigma_1^2 = 0.000784$. We analyzed various change-point locations, $\tau = 25, 50, 83, 125$, using $\mathcal{R} = 1000$ replications of each experiment. Figure \ref{s04r} provides a visualization of the cases where the change-point occurs at $\tau = 25$ and $\tau = 125$. In the Supplement, we explore another scenario with $\mu_0 = \mu_1 = 0.0008$, $\sigma_0^2 = 0.000169$, and $\sigma_1^2 = 0.000676$, considering various change-point locations and 1000 replications for each experiment, yielding very similar results.
 Detailed numerical results for these two scenarios are provided in  the Supplement.

\begin{figure}[H]    
	\centering   
	\includegraphics[width=8cm,height=4cm]{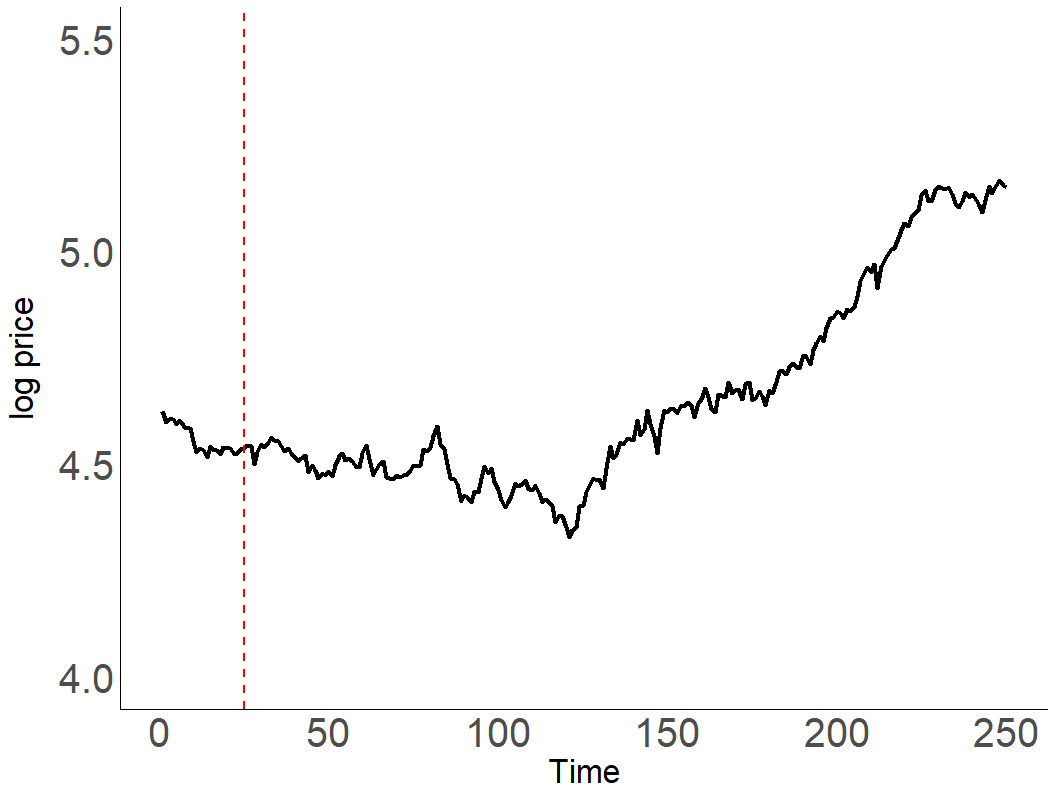}
	\includegraphics[width=8cm,height=4cm]{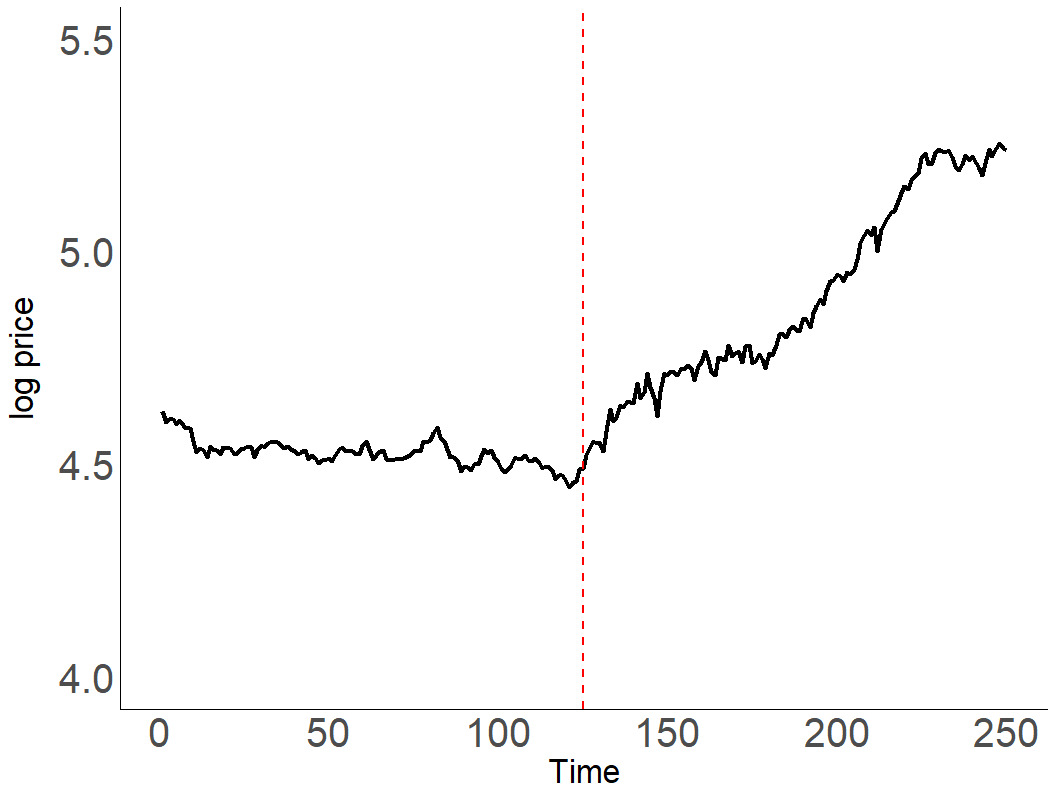}
	\caption{
		An illustration of two distinct change points, $\tau = 25$ (left) and $\tau = 125$ (right), indicated by the dashed lines, under the model parameters $\mu_0 = \mu_1 = 0.0008$, $\sigma_0^2 = 0.000169$, and $\sigma_1^2 = 0.000784$.}
	\label{s04r}
\end{figure}

Here, we report results for the first scenario involving structural changes in $\sigma^2$. Figure \ref{fig:sigma784} shows the RE and RMSE for all parameter estimates $(\widehat{\mu}_0, \widehat{\mu}_1, \widehat{\sigma}_0^2, \widehat{\sigma}_1^2, \widehat{\tau})$. We observe that the proposed OULC model consistently outperformed the traditional OC model, yielding both lower RE and RMSE across all change-point scenarios. In particular, in the challenging case where the change point occurs early at $\tau = 25$, the OULC model demonstrates a significant advantage. This advantage diminishes for the estimates of $(\widehat{\mu}_0, \widehat{\mu}_1, \widehat{\tau})$ as the change point moves to later stages. However, the advantage holds for $(\widehat{\sigma}_0^2, \widehat{\sigma}_1^2)$ at all change-point locations. Specifically, in the change point estimation of $\tau = 25$, the MLE, RMSE and RE in the traditional OC model are 29.704, 18.800053 and 0.752002, respectively, whereas in the OULC model, they are 25.034, 0.475395, and 0.019016, respectively.

Figure \ref{fig:sigma784} displays the averaged errors of individual estimations in all 1000 replications. Figure \ref{fig:B_sigma784} shows the overall mean estimates of the two models compared to the true values. Across all experiments, the OC model exhibits greater uncertainty in all parameter estimates.

\begin{figure}[htbp!]
	\centering
	\includegraphics[width=17cm,height=17cm]{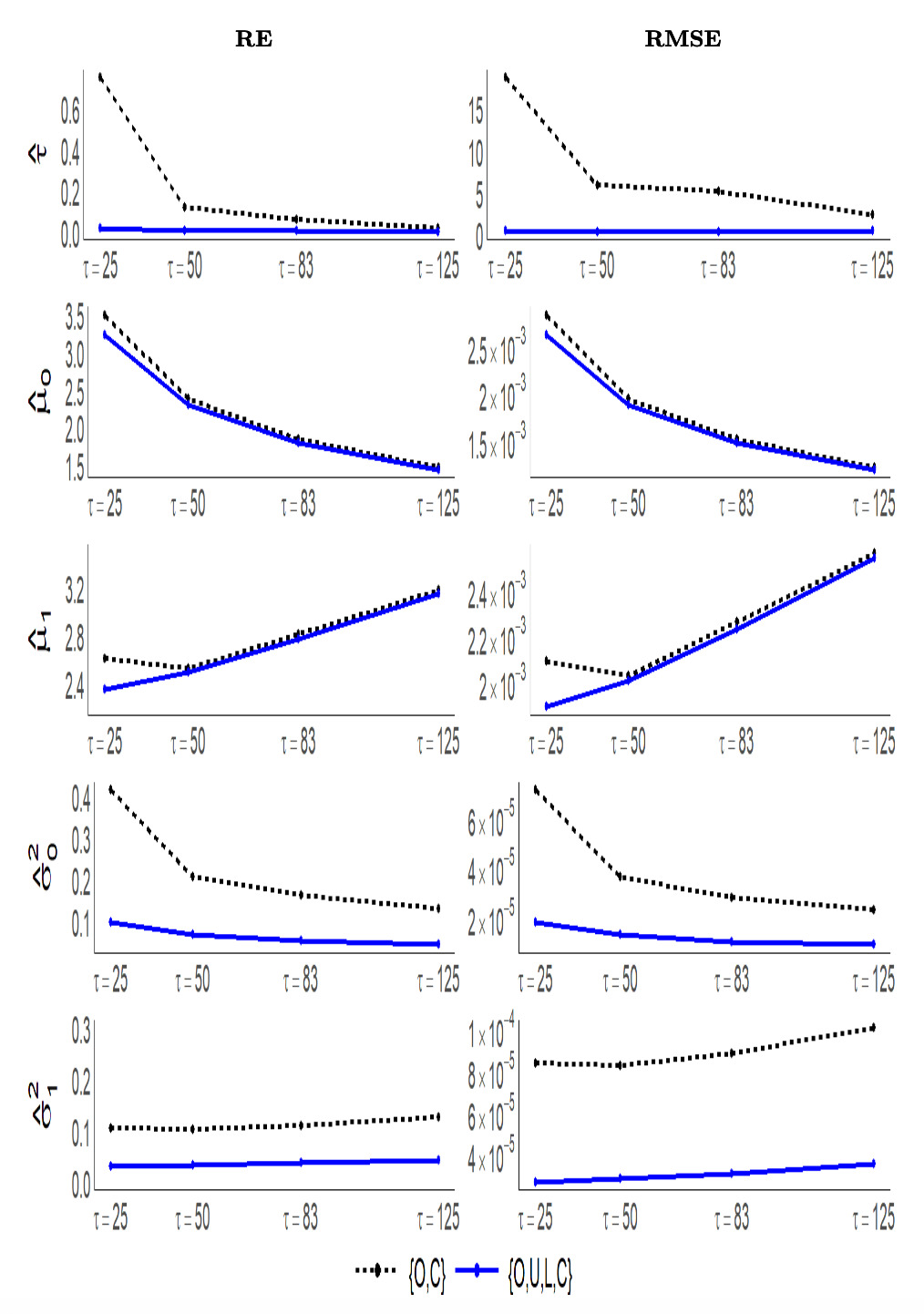}
	\caption{Performance analysis of the OC model and the OULC model in terms of RE and RMSE for the case $\mu_0 = \mu_1 = 0.0008$, $\sigma_0^2 = 0.000169$, and $\sigma_1^2 = 0.000784$, across varied $\tau$ values with 1000 replications.
	 }
	\label{fig:sigma784}
\end{figure}
\begin{figure}[htbp!]
	\centering
	\includegraphics[width=17cm,height=17cm]{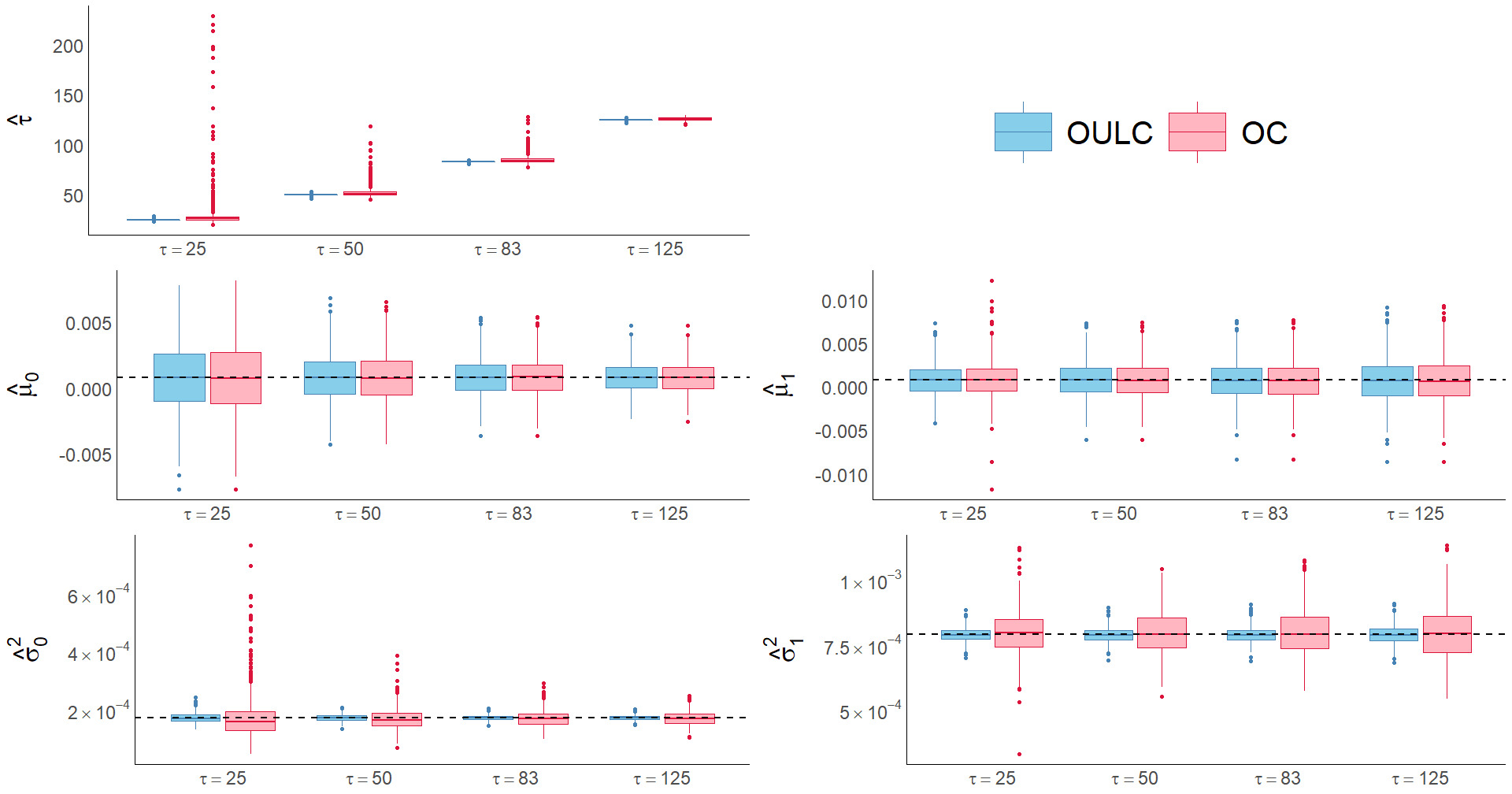}
	\caption{Boxplot of parameter estimations for the OC model and the OULC model for the case $\mu_0 = \mu_1 = 0.0008$, $\sigma_0^2 = 0.000169$, and $\sigma_1^2 = 0.000784$, across varied $\tau$ values with 1000 replications.
	}
	\label{fig:B_sigma784}
\end{figure}

%

\subsection{Structure change based on $\mu$}
\label{sec:Structure_change_mu}

\begin{figure}[H]    
	\centering   
	\includegraphics[width=8cm,height=4cm ]{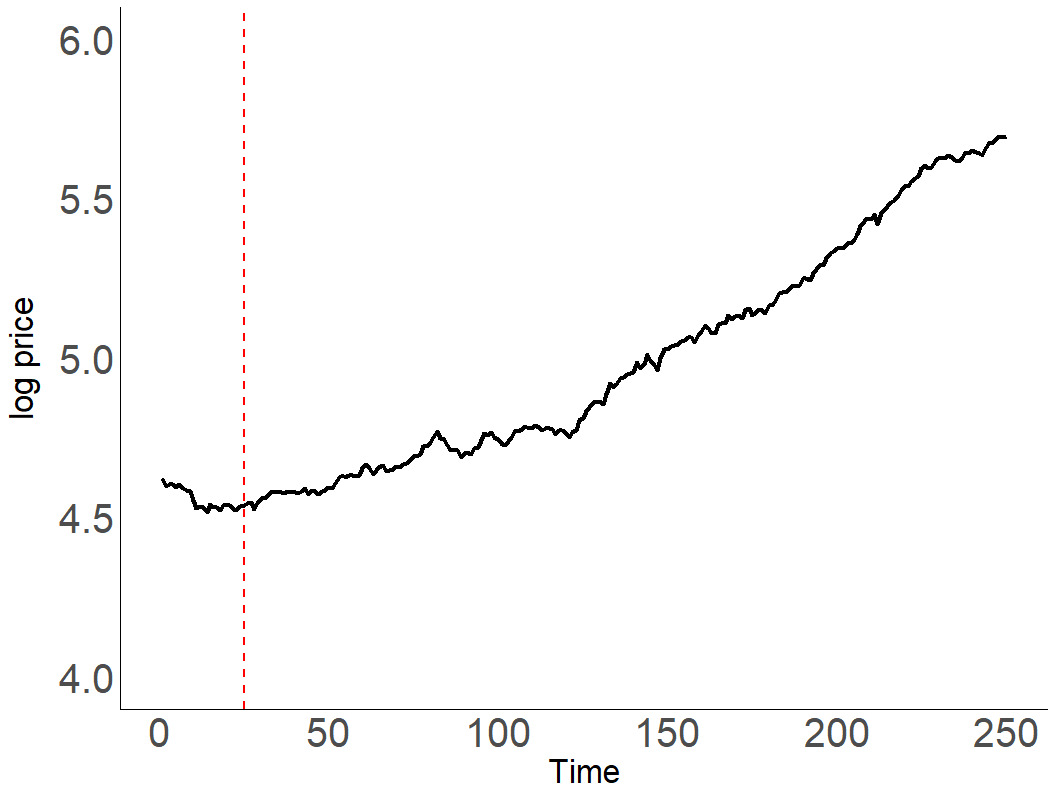}
	\includegraphics[width=8cm,height=4cm]{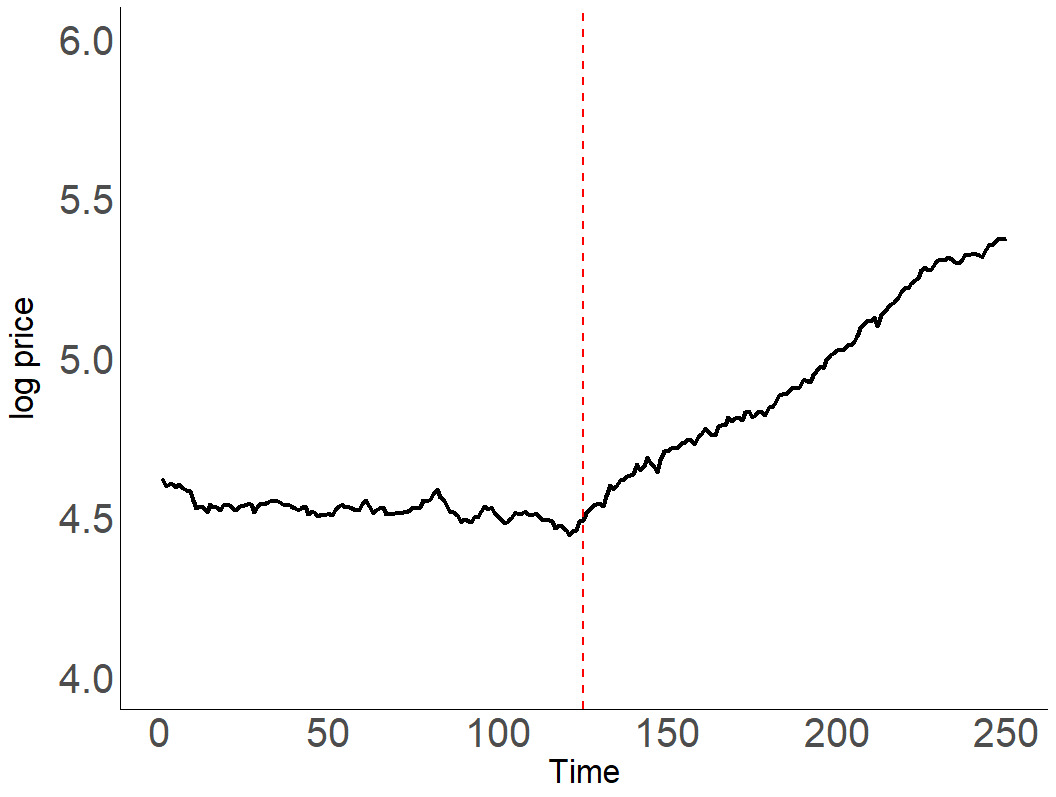}
	\caption{An illustration of two distinct change points, $\tau = 25$ (left) and $\tau = 125$ (right), indicated by the dashed lines, under the model parameters $\mu_0=0.0008$ and $\mu_1=0.004$ and $\sigma_0^2=\sigma_1^2=0.000169$.}
	\label{mu04r}
\end{figure}

\begin{figure}[htbp!]
	\centering
	\includegraphics[width=17cm,height=17cm]{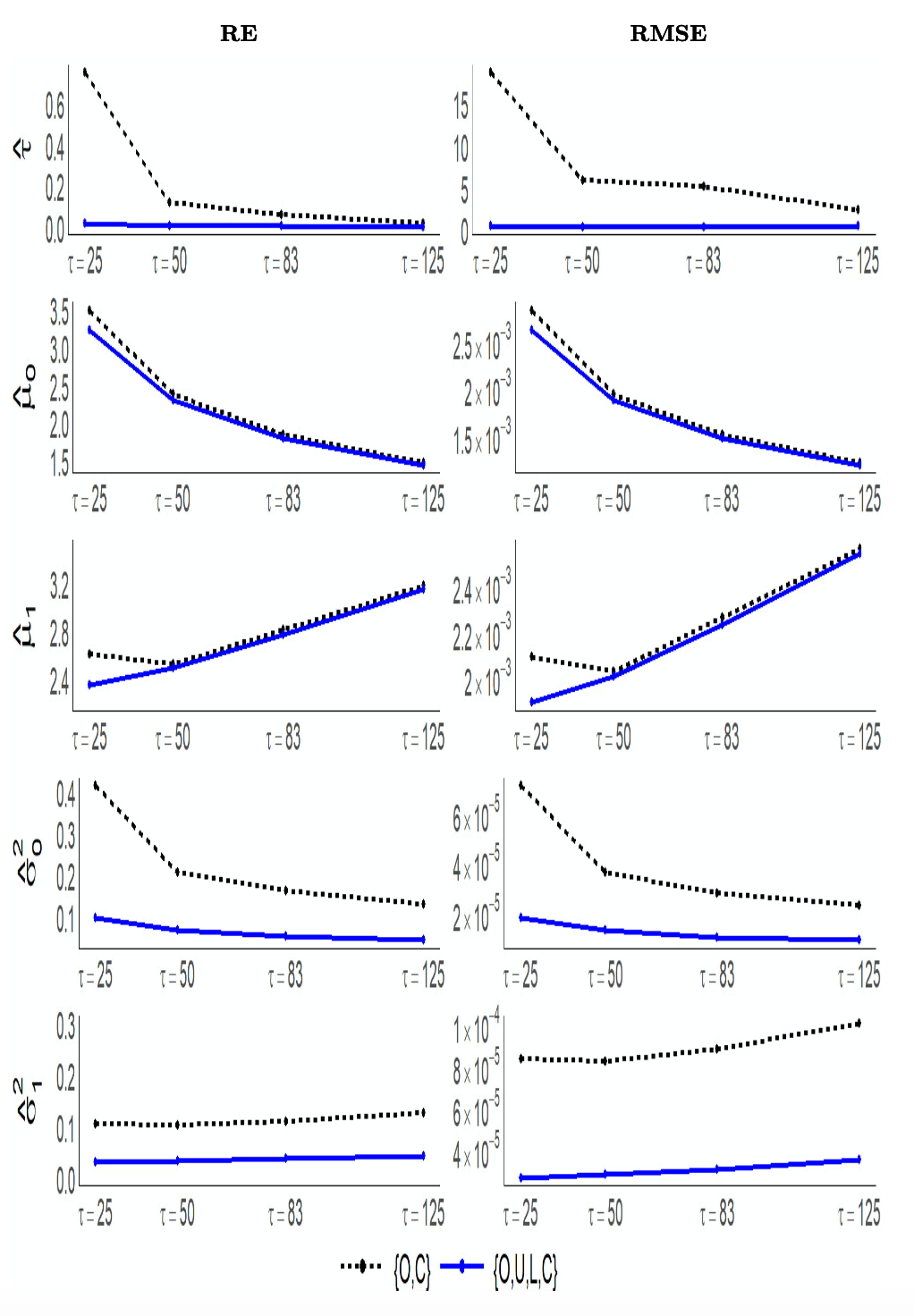}
	\caption{Performance analysis of the OC model and the OULC model in terms of RE and RMSE for the case $\mu_0=0.0008$ and $\mu_1=0.004$, $\sigma_0^2=\sigma_1^2=0.000169$, across varied $\tau$ values with 1000 replications.}
	\label{fig:mu5}
\end{figure}

\begin{figure}[htbp!]
	\centering
	\includegraphics[width=17cm,height=17cm]{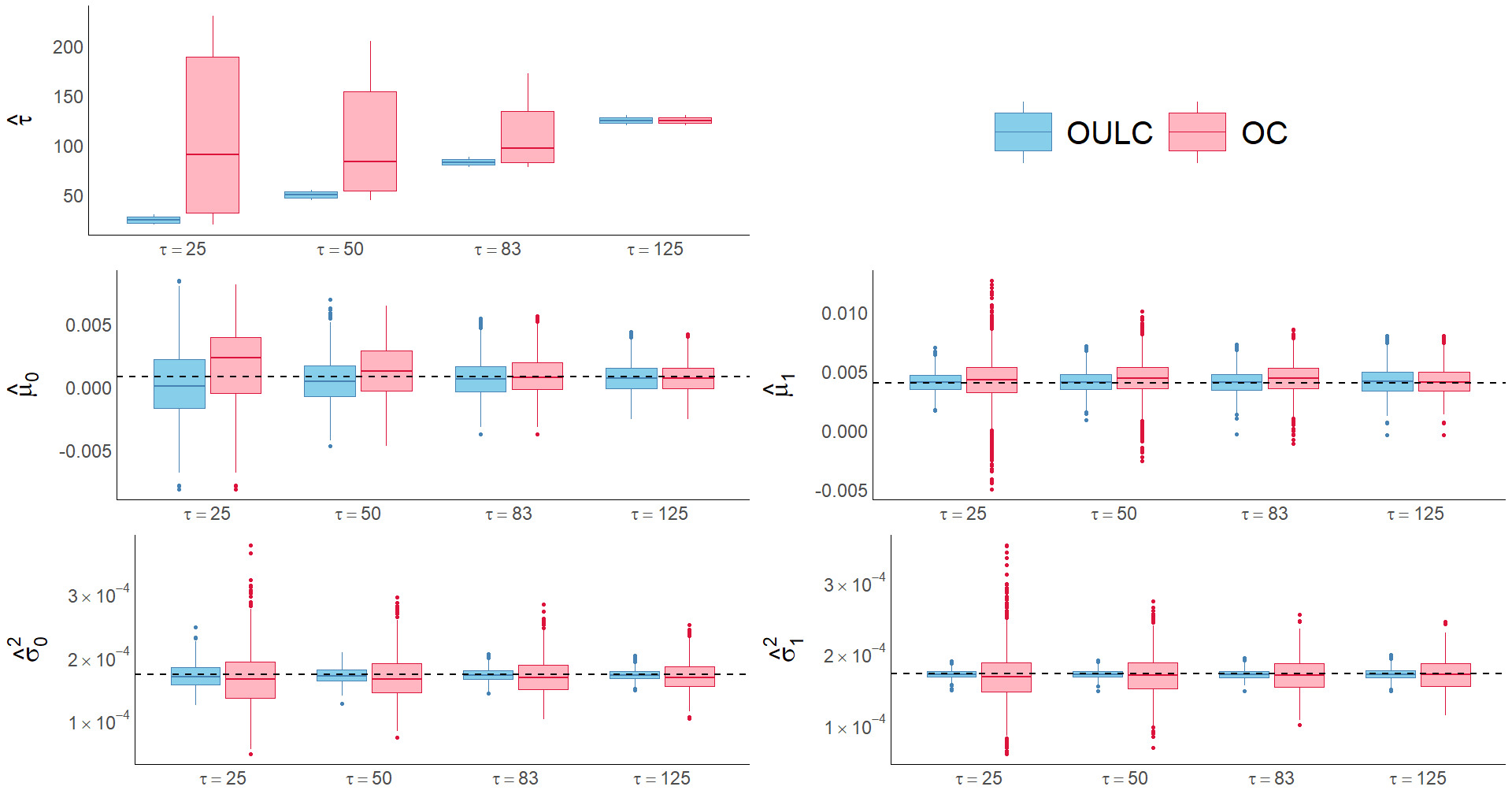}
	\caption{
		Boxplot of parameter estimations for the OC model and the OULC model for the case $\mu_0=0.0008$ and $\mu_1=0.004$, $\sigma_0^2=\sigma_1^2=0.000169$, across varied $\tau$ values with 1000 replications.}
	\label{fig:B_mu5}
\end{figure}
We assessed the model performance in the presence of structural changes in $\mu$. Initially, we focused on a change-point scenario defined by $\mu_0 = 0.0008$, $\mu_1 = 0.004$, and $\sigma_0^2 = \sigma_1^2 = 0.000169$. Various change-point locations, specifically $\tau = 25, 50, 83, 125$, were analyzed using $\mathcal{R} = 1000$ replications for each experiment. A visualization of the change-point occurrences at $\tau = 25$ and $\tau = 125$ is presented in Figure \ref{mu04r}. Detailed numerical results for both scenarios can be found in  the Supplement.

We present the results for the first scenario characterized by structural changes in $\mu$. Figure \ref{fig:mu5} illustrates the RE and RMSE for all parameter estimates $(\widehat{\mu}_0, \widehat{\mu}_1, \widehat{\sigma}_0^2, \widehat{\sigma}_1^2, \widehat{\tau})$. The findings reveal that the proposed OULC model consistently outperforms the traditional OC model, exhibiting lower RE and RMSE across all change-point scenarios.  Similar to the previous case in Figure \ref{fig:sigma784}, the OC model struggles to estimate $(\widehat{\sigma}_0^2, \widehat{\sigma}_1^2)$ despite these parameters remaining unchanged in this scenario. In particular, in the change point estimation of $\tau = 25$, the traditional OC model yields an estimate $\widehat{\tau} = 109.063$ with RMSE 113.959620 and RE 4.558385; in contrast, the proposed OULC model provides a much more accurate estimate $\widehat{\tau} = 24.876$ with RMSE 3.532138 and 0.141286 .

Figure \ref{fig:B_mu5} shows the overall mean estimates of the two models compared to the true values. Across all experiments, the OC model exhibits greater uncertainty in all parameter estimates. Unlike Figure \ref{fig:B_sigma784}, where the mean of the estimates generated by the OC model was close to the true value, in this case study the OC model's estimates for the change point $\widehat{\tau}$ deviate significantly from the true value and display large variations. In contrast, the OULC model consistently provides stable and accurate estimates across all experiments.


	\section{Empirical Data analysis}\label{sec:ES}
In the empirical analysis of the S\&P 500 index during the Russo-Ukrainian War in 2022, the proposed OULC model demonstrates superior performance compared to the traditional OC model. The datasets, acquired from Yahoo Finance (\url{finance.yahoo.com}), cover the daily log prices from December 31, 2021, to May 20, 2022. Notably, from Table \ref{2022}, we can see that the OULC model yields a much smaller Akaike Information Criterion (AIC) value of -2207.14, compared to -565.36 for the OC model, indicating a much better fit to the data. Furthermore, the confidence interval for the change point $\tau$ identified by the OULC model is narrower, suggesting increased precision in the estimation of the change point, which is determined to be April 19, 2022. In contrast, the OC model presents a much wider 95\% confidence interval for $\tau$ at 75 being (4, 93), highlighting the OULC model's enhanced reliability in capturing structural changes. These findings illustrate that the OULC model not only provides more accurate parameter estimates but also offers a more convincing representation of the dynamics surrounding significant financial events, as visualized by the real data plot in Figure \ref{RUwar}. 
\begin{table}[H]
	\centering
	\small
	\begin{tabular}{@{}lll@{}}
		\toprule
		Model & \{O,U,L,C\} & \{O, C\} \\ \midrule
		
		$\tau$ - Date & 2022-04-19 & 2022-04-20 \\ 
		$\tau$ - 95\%CI  & 74  (70, 79) & 75  (4, 93) \\ 
		$\widehat{\mu}_0$ - 95\%CI  & -0.0008136 (-0.0030776, 0.0015171) & -0.0006172 (-0.0117505, 0.0085377) \\ 
		$\widehat{\mu}_1$ - 95\%CI  &  -0.0044948  (-0.0096793, 0.0008978) & -0.0055563 (-0.0169120, 0.0111408) \\ 
		$\widehat{\sigma}^2_0$ - 95\%CI  & 0.0001069 (0.0000958, 0.0001170) & 0.0001413 (0.0000080, 0.0002002) \\ 
		$\widehat{\sigma}^2_1$ - 95\%CI  & 0.0001956 (0.0001605, 0.0002347) & 0.0002915 (0.0000182, 0.0005057) \\ 
		AIC  & -2207.14 & -565.36\\
		\bottomrule
	\end{tabular}
	\caption{Parameter estimates and model performance metrics for the OULC and OC models based on the S\&P 500 index during the Russo-Ukrainian War in 2022.}
	\label{2022}
\end{table}

\begin{figure}[H]  
	\centering
	\includegraphics[width=0.9\textwidth]{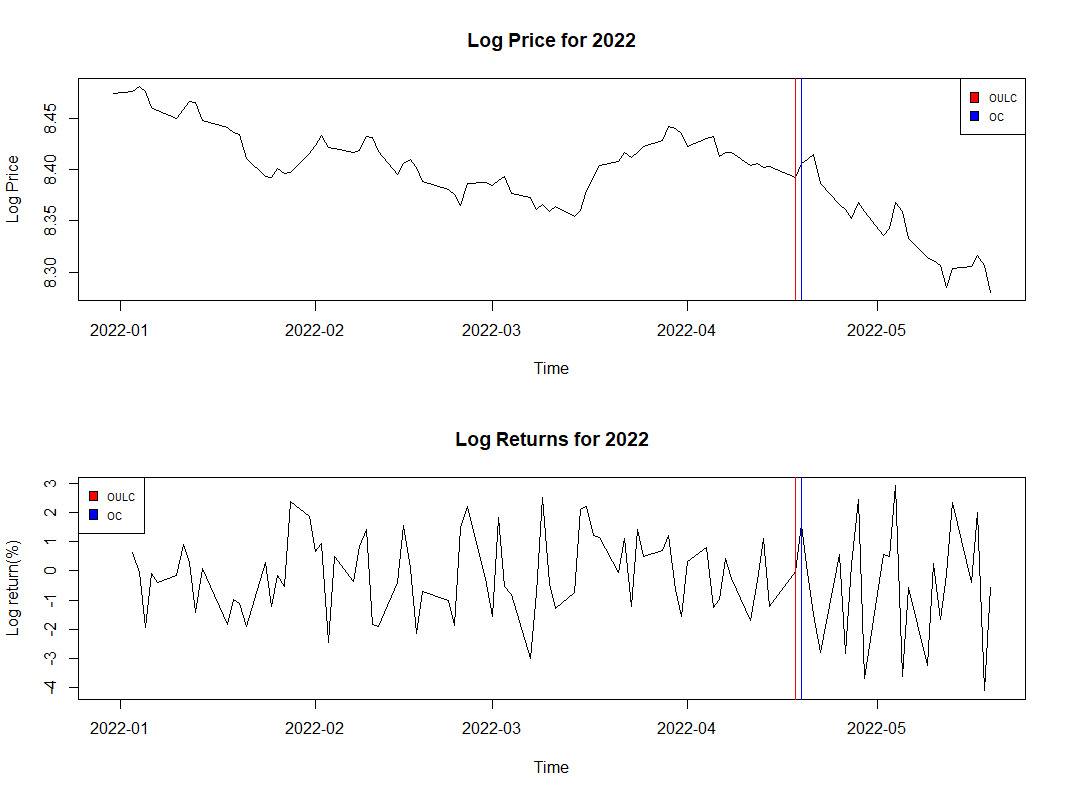}
	\caption{Change-point analysis for the S\&P 500 index during the Russo-Ukrainian War in 2022.}
	\label{RUwar}
\end{figure}

	\section{Conclusions}\label{CON}
In this paper, we have introduced a novel approach for change-point estimation in interval-based time series, leveraging the GBM model alongside the Girsanov theorem to capture multivariate time series data, including daily maximum, minimum, opening, and closing prices. The estimation framework utilizes MLE in conjunction with the NR algorithm, delivering robust performance in both simulated and real-world scenarios. Our simulation studies demonstrate that the proposed method achieves high accuracy in estimating change-points under varying parameter settings for mean ($\mu$) and variance ($\sigma^2$). Comparisons with a model based solely on closing prices confirm that our method outperforms it in terms of RMSE and RE. In empirical applications, our approach successfully detects critical change-points associated with the Russo-Ukrainian War in 2022.


	\section*{Funding} 
The research of L.S. was supported by NSTC grant 112-2628-M-008-002-MY3. The research of N.N. was partially supported by NIH grant 1R21AI180492-01 and the Individual Research Grant at Texas A\&M University.

\bibliography{sn-offline-GBM}
\end{document}